\documentclass[runningheads]{llncs}
\usepackage{graphicx}
\usepackage{enumerate}
\usepackage[all]{xy}
\usepackage[svgnames]{xcolor}
\usepackage{amsmath,amssymb,wasysym}
\usepackage{mathtools}
\usepackage{tikz}
\usepackage{tikz-cd}
\newcommand{\cact}[1]{\overline{#1}}
\newcommand{\D}[0]{\mathsf{D}}
\newcommand{\dd}[0]{{\boldsymbol \partial}}
\newcommand{\pair}[2]{\left\langle {#1}, {#2} \right\rangle}
\newcommand{\four}[4]{\pair{\pair{#1}{#2}}{\pair{#3}{#4}}}

\newcommand{\CdCax}[1]{{[\bf C$\dd$.{#1}]}}

\newcommand{\red}[1]{{\color{red} {#1}}}

\newcommand{\blue}[1]{{\color{blue} {#1}}}
\newcommand{\js}[1]{ \blue{ JS: {#1}} }

\begin{document}
\title{
    Cartesian Difference Categories
}
\subtitle{Extended Report \footnote{This is an extended and corrected version of the FOSSACS conference paper \cite{alvarez2020cartesian}.}}

\author{Mario Alvarez-Picallo\inst{1}\and
Jean-Simon Pacaud Lemay\inst{2}}
\authorrunning{M. Alvarez-Picallo \& J-S. P. Lemay}
\institute{Department of Computer Science, University of Oxford, UK \\
\email{mario.alvarez-picallo@cs.ox.ac.uk} \and
Department of Computer Science, University of Oxford, UK \\
\email{jean-simon.lemay@kellogg.ox.ac.uk}}
\maketitle              % typeset the header of the contribution
\begin{abstract}
Cartesian differential categories are categories equipped with a differential combinator which axiomatizes the directional derivative. Important models of Cartesian differential categories include classical differential calculus of smooth functions and categorical models of the differential $\lambda$-calculus. However, Cartesian differential categories cannot account for other interesting notions of differentiation such as the calculus of finite differences or the Boolean differential calculus. On the other hand, change action models have been shown to capture these examples as well as more ``exotic'' examples of differentiation. However, change action models are very general and do not share the nice properties of a Cartesian differential category. In this paper, we introduce Cartesian difference categories as a bridge between Cartesian differential categories and change action models. We show that every Cartesian differential category is a Cartesian difference category, and how certain well-behaved change action models are Cartesian difference categories. In particular, Cartesian difference categories model both the differential calculus of smooth functions and the calculus of finite differences. Furthermore, every Cartesian difference category comes equipped with a tangent bundle monad whose Kleisli category is again a Cartesian difference category. 

\end{abstract}
\section{Introduction}

In the early 2000's, Ehrhard and Regnier introduced the differential $\lambda$-calculus \cite{ehrhard2003differential}, an extension of the $\lambda$-calculus equipped with a differential combinator capable of taking the derivative of arbitrary higher-order functions. This development, based on models of linear logic equipped with a natural notion of ``derivative'' \cite{ehrhard2018introduction}, sparked a wave of research into categorical models of differentiation. 

One of the most notable developments in the area are Cartesian differential categories \cite{blute2009cartesian}, introduced by Blute, Cockett and Seely, which provide an abstract categorical axiomatization of the directional derivative from differential calculus. The relevance of Cartesian differential categories lies in their ability to model both ``classical'' differential calculus (with the canonical example being the category of Euclidean spaces and smooth functions between) and the differential $\lambda$-calculus (as every categorical model for it gives rise to a Cartesian differential category \cite{manzonetto2012categorical}). However, while Cartesian differential categories have proven to be an immensely successful formalism, they have, by design, some limitations. Firstly, they cannot account for certain ``exotic'' notions of derivative, such as the difference operator from the calculus of finite differences \cite{richardson1954introduction} or the Boolean differential calculus \cite{steinbach2009boolean}. This is because the axioms of a Cartesian differential category stipulate that derivatives should be linear in their second argument (in the same way that the directional derivative is), whereas these aforementioned discrete sorts of derivative need not be. Additionally, every Cartesian differential category is equipped with a tangent bundle monad \cite{cockett2014differential,manzyuk2012tangent} whose Kleisli category can be intuitively understood as a category of generalized vector fields. %This Kleisli category has an obvious differentiation operator which comes close to making it a Cartesian differential category, but again fails the requirement of being linear in its second argument. 

More recently, discrete derivatives have been suggested as a semantic framework for understanding incremental computation. This led to the development of change structures \cite{cai2014theory} and change actions \cite{alvarez2019change}. Change action models have been successfully used to provide a model for incrementalizing Datalog programs \cite{alvarez2019fixing}, and have also been shown to model the calculus of finite differences.
% as well as the Kleisli category of the tangent bundle monad of a Cartesian differential category. 
Change action models, however, are very general, lacking many of the nice properties of Cartesian differential categories (for example, addition in a change action model is not required to be commutative), even though they are verified in most change action models. As a consequence of this generality, the tangent bundle endofunctor in a change action model can fail to be a monad. 

In this work we introduce Cartesian difference categories (Section \ref{CdCsec}), whose key ingredients are an infinitesimal extension operator and a difference combinator, whose axioms are a generalization of the differential combinator axioms of a Cartesian differential category. In Section \ref{CDCisCdCsec}, we show that every Cartesian differential category is in fact a Cartesian difference category whose infinitesimal extension operator is zero, and conversely how every Cartesian difference category admits a full subcategory which is a Cartesian differential category. In Section \ref{CdCisCAsec}, we show that every Cartesian difference category is a change action model, and conversely how a full subcategory of suitably well-behaved objects of a change action model is a Cartesian difference category. In Section \ref{monadsec}, we show that every Cartesian difference category comes equipped with a monad whose Kleisli category again a Cartesian difference category. Finally, in Section \ref{EXsec} we provide some examples of Cartesian difference categories; notably, the calculus of finite differences and the stream calculus.  

\section{Cartesian Differential Categories}\label{CDCsec}

In this section we briefly review Cartesian differential categories, so that the reader may compare Cartesian differential categories with the new notion of Cartesian \emph{difference} categories which we introduce in the next section. For a full detailed introduction on Cartesian differential categories, we refer the reader to the original paper \cite{blute2009cartesian}. 

\subsection{Cartesian Left Additive Categories}

Here we provide the definition of Cartesian left additive categories \cite{blute2009cartesian} -- where ``additive'' is meant being skew enriched over commutative monoids. In particular this means that we do not assume the existence of additive inverses, i.e., ``negative elements''. 

\begin{definition}\label{LACdef} A \textbf{left additive category} \cite{blute2009cartesian} is a category $\mathbb{X}$ such that each hom-set $\mathbb{X}(A,B)$ is a commutative monoid with addition operation $+: \mathbb{X}(A,B) \times \mathbb{X}(A,B) \to \mathbb{X}(A,B)$ and zero element (called the zero map) $0 \in \mathbb{X}(A,B)$, such that pre-composition preserves the additive structure: $(f + g) \circ h = f \circ h + g \circ h$ and $0 \circ f = 0$. A map $k$ in a left additive category is \textbf{additive} if post-composition by $k$ preserves the additive structure: $k \circ (f +g) = k \circ f + k \circ g$ and $k \circ 0=0$. 
\end{definition} 

By a Cartesian category we mean a category $\mathbb{X}$ with chosen finite products where we denote the binary product of objects $A$ and $B$ by $A \times B$ with projection maps $\pi_0: A \times B \to A$ and $\pi_1: A \times B \to B$ and pairing operation $\langle -, - \rangle$, and the chosen terminal object as $\top$ with unique terminal maps $!_A: A \to \top$. 

\begin{definition}\label{CLACdef}
A \textbf{Cartesian left additive category} \cite{blute2009cartesian} is a Cartesian category $\mathbb{X}$ which is also a left additive category such all projection maps $\pi_0: A \times B \to A$ and $\pi_1: A \times B \to B$ are additive \footnote{Note that the definition given here of a Cartesian left additive category is slightly different from the one found in \cite{blute2009cartesian}, but it is indeed equivalent.}. 
\end{definition}

By \cite[Proposition 1.2.2]{blute2009cartesian}, an equivalent axiomatization is of a Cartesian left additive category is that of a Cartesian category where every object comes equipped with a commutative monoid structure such that the projection maps are monoid morphisms. This will be important later in Section \ref{CdCsec}. 

\subsection{Cartesian Differential Categories}

\begin{definition}\label{cartdiffdef} A \textbf{Cartesian differential category} \cite{blute2009cartesian} is a Cartesian left additive category equipped with a \textbf{differential combinator} $\mathsf{D}$ of the form\footnote{Note that here, following the more recent work on Cartesian differential categories, we've flipped the convention found in \cite{blute2009cartesian}, so that the linear argument is in the second argument rather than in the first argument. }: 
\[ \frac{f : A \to B}{\mathsf{D}[f]: A \times A \to B} \]
 verifying the following coherence conditions: 
\begin{enumerate}[{\bf [CD.1]}]
\item $\mathsf{D}[f+g] = \mathsf{D}[f] + \mathsf{D}[g]$ and $\mathsf{D}[0]=0$
\item $\mathsf{D}[f] \circ \langle x, y+z \rangle= \mathsf{D}[f] \circ \langle x, y \rangle + \mathsf{D}[f] \circ \langle x, z \rangle$ and $\mathsf{D}[f] \circ \langle x, 0 \rangle=0$
\item $\mathsf{D}[1_A]=\pi_1$ and $\mathsf{D}[\pi_0] = \pi_0 \circ \pi_1$ and $\mathsf{D}[\pi_1] = \pi_1 \circ \pi_1$ 
\item $\mathsf{D}[\langle f, g \rangle] = \langle  \mathsf{D}[f], \mathsf{D}[g] \rangle$ and $\mathsf{D}[!_A]=!_{A \times A}$
\item $\mathsf{D}[g \circ f] = \mathsf{D}[g] \circ \langle f \circ \pi_0, \mathsf{D}[f] \rangle$
\item $\mathsf{D}\left[\mathsf{D}[f] \right] \circ \left \langle \langle x, y \rangle, \langle 0, z \rangle \right \rangle=  \mathsf{D}[f] \circ \langle x, z \rangle$
\item $\mathsf{D}\left[\mathsf{D}[f] \right] \circ \left \langle \langle x, y \rangle, \langle z, 0 \rangle \right \rangle = \mathsf{D}\left[\mathsf{D}[f] \right] \circ \left \langle \langle x, z \rangle, \langle y, 0 \rangle \right \rangle$
\end{enumerate}
\end{definition}

We highlight that by \cite[Proposition 4.2]{cockett2014differential}, the last two axioms {\bf [CD.6]} and {\bf [CD.7]} have an equivalent alternative expression. 
\begin{lemma}In the presence of the other axioms, {\bf [CD.6]} and {\bf [CD.7]} are equivalent to:
 \begin{enumerate}[{\bf [CD.1.{a}]}]
 \setcounter{enumi}{5}
\item $\mathsf{D}\left[\mathsf{D}[f] \right] \circ \left \langle \langle x, 0 \rangle, \langle 0, y \rangle \right \rangle= \mathsf{D}[f] \circ  \langle x, y \rangle$
\item $\mathsf{D}\left[\mathsf{D}[f] \right] \circ \left \langle \langle x, y \rangle, \langle z, w \rangle \right \rangle= \mathsf{D}\left[\mathsf{D}[f] \right] \circ \left \langle \langle x, z \rangle, \langle y, w \rangle \right \rangle$
\end{enumerate}
\end{lemma}

As a Cartesian difference category is a generalization of a Cartesian differential category, we leave the discussion of the intuition of these axioms for later in Section \ref{CdCsec} below. We also refer to \cite[Section 4]{blute2009cartesian} for a term calculus which may help better understand the axioms of a Cartesian differential category. The canonical example of a Cartesian differential category is the category of real smooth functions, which we will discuss in Section \ref{smoothex}. Other interesting examples of can be found throughout the literature such as categorical models of the differential $\lambda$-calculus \cite{ehrhard2003differential,manzonetto2012categorical}, the subcategory of differential objects of a tangent category \cite{cockett2014differential}, and the coKleisli category of a differential category \cite{blute2006differential,blute2009cartesian}.

\section{Change Action Models}\label{CAsec}

Change actions \cite{alvarez2019fixing,alvarez2019change} have recently been proposed as a setting for reasoning about higher-order incremental computation, based on a discrete notion of differentiation. Together with Cartesian differential categories, they provide the core ideas behind Cartesian difference categories. In this section we quickly review change actions and change action models, in particular to highlight where the some of the axioms of a Cartesian difference category come from. For more details on change actions, we invite readers to see the original paper \cite{alvarez2019change}. 

\subsection{Change Actions}
\begin{definition}
    A \textbf{change action} $\cact{A}$ in a Cartesian category $\mathbb{X}$ is a quintuple $\cact{A} \equiv (A, \Delta A, \oplus_A, +_A, 0_A)$ consisting of two objects $A$ and $\Delta A$, and three maps:
\[ \oplus_A: A \times \Delta A \to A \quad \quad \quad +_A: \Delta A \times \Delta A \to \Delta A \quad \quad \quad 0_A: \top \to \Delta A \]
such that $(\Delta A, +_A, 0_A)$ is a monoid and $\oplus_A : A \times \Delta A \to A$ is an action of $\Delta A$ on $A$, that is, the following equalities hold:
    \[ \oplus_A \circ \langle 1_A, 0_A \circ !_A \rangle = 1_A \quad \quad \quad \oplus_A \circ (1_A \times +_A) = \oplus_A \circ (\oplus_A \times 1_{\Delta A}) \]
        \iffalse
        \begin{center}
        \begin{tikzcd}[column sep=7em]
            A \times \Delta A \times \Delta A
            \arrow[r, "1_A \times +_A"]
            \arrow[d, "\oplus_A \times 1_{\Delta A}"]
            &
            A \times \Delta A
            \arrow[d, "\oplus_A"]
            \\
            A \times \Delta A
            \arrow[r, "\oplus_A"]
            &
            A
        \end{tikzcd}
        \end{center}\fi
\end{definition}

For a change action $\cact{A}$ and given a pair of maps $f: C \to A$ and $g: C \to \Delta A$, we define $f \oplus_{\cact{A}} g: C \to A$ as $f \oplus_{\cact{A}} g = \oplus_A \circ \langle f, g \rangle$. Similarly, for maps $h: C \to \Delta A$ and $k: C \to \Delta A$, define $h +_{\cact{A}} k = +_A \circ \langle h, k \rangle$. Therefore, that $\oplus_A$ is an action of $\Delta A$ on $A$ can be rewritten as: 
\[ 1_A \oplus_{\cact{A}} 0_A = 1_A \quad \quad \quad 1_A \oplus_{\cact{A}} (1_{\Delta A} +_{\cact{A}} 1_{\Delta A}) = (1_A \oplus_{\cact{A}} 1_{\Delta A}) \oplus_{\cact{A}} 1_{\Delta A} \]

The intuition behind the above definition is that the monoid $\Delta A$ is a type of possible ``changes'' or ``updates'' that might be applied to $A$, with the monoid structure on $\Delta A$ representing the capability to compose updates.

\begin{example}
    Any monoid $(A, +_A, 0_A)$, where $+_A: A \times A \to A$ and $0_A: \top \to A$, acts on itself on the right and therefore defines a change action $(A, A, +_A, +_A, 0_A)$.
\end{example}
\begin{example}
  In a Cartesian closed category, with internal hom $- \Rightarrow -$, evaluation map $\mathbf{ev}$, and currying operator $\Lambda(-)$, the quintuple $(A, A \Rightarrow A, \mathbf{ev}, \circ, \Lambda(1_A))$ is a change action.
\end{example}

Change actions give rise to a notion of derivative, with a distinctly ``discrete'' flavour. Given change actions on objects $A$ and  $B$, a map $f : A \to B$ can be said to be differentiable when changes to the input (in the sense of elements of $\Delta A$) are mapped to changes to the output (that is, elements of $\Delta B$). In the setting of incremental computation, this is precisely what it means for $f$ to be incrementalizable, with the derivative of $f$ corresponding to an incremental version of $f$.

\begin{definition}\label{def:derivative}
    Let $\cact{A} \equiv (A, \Delta A, \oplus_A, +_A, 0_A)$ and $\cact{B} \equiv
    (B, \Delta B, \oplus_B, +_B, 0_B)$ be change actions. For a map $f: A \to
    B$, a map $\dd[f] : A \times \Delta A \to \Delta B$ is a \textbf{derivative}
    of $f$ whenever the following axioms hold: 
    \begin{enumerate}[{\bf [C{A}D.1]}]
        \item $f \circ (x \oplus_{\cact{A}} y) = f \circ x \oplus_{\cact{B}} (\dd[f] \circ \pair{x}{y})$
        \item $\dd[f] \circ \pair{x}{y +_{\cact{A}} z} = (\dd[f] \circ \pair{x}{y}) +_{\cact{B}} (\dd[f] \circ \pair{x \oplus_{\cact{A}} y}{z})$ and \\ \noindent$\dd[f] \circ \langle x, 0_B \circ !_B \rangle = 0_B \circ !_{A \times \Delta A}$
    \end{enumerate}
\end{definition}

The intuition for these axioms will be explain in more details in Section
\ref{CdCsec} when we detail the axioms of a Cartesian difference category. Note
that although there is nothing in the above definition guaranteeing that any
given map have at most a single derivative, the chain rule does hold. As a
corollary, differentiation is compositional and therefore the change actions in
$\mathbb{X}$ form a category. 

\begin{lemma} Whenever $\dd[f]$ and $\dd[g]$ are derivatives for composable maps
$f$ and $g$ respectively, then $\dd[g] \circ \langle f \circ \pi_0, \dd[f]
\rangle$ is a derivative for $g \circ f$.
\end{lemma}

\subsection{Change Action Models}

\begin{definition}
    Given a Cartesian category $\mathbb{X}$, define its change actions category
    $\mathsf{CAct}(\mathbb{X})$ as the category whose objects are change actions
    in $\mathbb{X}$ and whose arrows $\cact{f} : \cact{A} \to \cact{B}$ are the
    pairs $(f, \dd[f])$, where $f : A \to B$ is an arrow in $\mathbb{X}$ and
    $\dd[f] : A \times \Delta A \to \Delta B$ is a derivative for $f$. The
    identity is $(1_A, \pi_1)$, while composition of $(f, \dd[f])$ and $(g,
    \dd[g])$ is $(g \circ f, \dd[g] \circ \langle f \circ \pi_0, \dd[f]
    \rangle)$. 
\end{definition}

There is an obvious product-preserving forgetful functor $\mathcal{E} :
\mathsf{CAct}(\mathbb{X}) \to \mathbb{X}$ sending every change action $(A,
\Delta A, \oplus, +, 0)$ to its base object $A$ and every map $(f, \dd[f])$ to
the underlying map $f$. As a setting for studying differentiation, the category
$\mathsf{CAct}(\mathbb{X})$ is rather lacklustre, since there is no notion of
higher derivatives, so we will instead work with change action models.
Informally, a change action model consists of a rule that associates, to every
object $A$ of $\mathbb{X}$, a change action over it, and to every map a choice
of derivative for it.

\begin{definition}
    A \textbf{change action model} is a Cartesian category $\mathbb{X}$ is a
    product-preserving functor $\alpha : \mathbb{X} \to
    \mathsf{CAct}(\mathbb{X})$ that is a section of the forgetful functor
    $\mathcal{E}$.
\end{definition}

For brevity, when $A$ is an object of a change action model, we will write
$\Delta A$, $\oplus_A$, $+_A$, and $0_A$ to refer to the components of the
corresponding change action $\alpha(A)$. Examples of change action models can be
found in \cite{alvarez2019change}. In particular, we highlight that a Cartesian
differential category always provides a change model action. We will generalize
this result, and show in Section \ref{CdCisCAsec} that a Cartesian difference
category also always provides a change action model. 

\section{Cartesian Difference Categories}
\newcommand{\RR}[0]{\mathbb{R}}

In this section we introduce \emph{Cartesian difference categories}, which are
generalizations of Cartesian differential categories equipped with an operator
that turns a map into an ``infinitesimal'' version of it. To give a clearer
understanding of this idea, consider the Taylor approximation of a
differentiable function $f : \RR \to \RR$ which states that, for any $x$ and
sufficiently small $y$:
\[
    f(x + y) \approx f(x) + y \cdot f'(x)
\]
The intuition behind the ``infinitesimal extension'' that we seek to define is
that it captures what is meant for $y$ to be ``sufficiently small''. 

\subsection{Infinitesimal Extensions in Left Additive Categories}

 \begin{definition} A Cartesian left additive category $\mathbb{X}$ is said to have an \textbf{infinitesimal extension} $\varepsilon$ if every homset $\mathbb{X}(A,B)$ comes equipped with a monoid morphism $\varepsilon: \mathbb{X}(A,B) \to \mathbb{X}(A,B)$, that is, $\varepsilon(f+g) = \varepsilon(f) + \varepsilon(g)$ and $\varepsilon(0) = 0$, and such that $\varepsilon(g \circ f) = \varepsilon(g) \circ f$ and $\varepsilon(\pi_0) = \pi_0 \circ \varepsilon(1_{A\times B})$ and $\varepsilon(\pi_1) = \pi_1 \circ \varepsilon(1_{A\times B})$. 
\end{definition}  

Note that since $\varepsilon(g \circ f) = \varepsilon(g) \circ f$, it follows
that $\varepsilon(f) = \varepsilon(1_B) \circ f$ and $\varepsilon(1_A)$ is an
additive map (Definition \ref{LACdef}). Furthermore, infinitesimal extensions
give rise to a canonical change action structure on each object:

\begin{lemma}\label{caelem} Let $\mathbb{X}$ be a Cartesian left additive category with infinitesimal extension $\varepsilon$. For every object $A$, define the maps $\oplus_A: A \times A \to A$, $+_A: A \times A \to A$, and $0_A: \top \to A$ respectively as follows: 
\begin{align*}
\oplus_A = \pi_0 + \varepsilon(\pi_1) && +_A = \pi_0 + \pi_1 && 0_A = 0 
\end{align*}
Then $(A, A, \oplus_A, +_A, 0_A)$ is a change action in $\mathbb{X}$. 
\end{lemma}
\begin{proof} As mentioned earlier, that $(A, +_A, 0_A)$ is a commutative monoid was shown in \cite{blute2009cartesian}. On the other hand, that $\oplus_A$ is an action follows from the fact that $\varepsilon$ preserves addition.  \hfill $\blacksquare$
\end{proof} 

Setting $\cact{A} \equiv (A, A, \oplus_A, +_A, 0_A)$, we note that $f
\oplus_{\cact{A}} g = f + \varepsilon(g)$ and $f +_{\cact{A}} g = f + g$, and so
in particular $+_{\cact{A}} = +$. Therefore, from now on we will omit the
subscripts and simply write $\oplus$ and $+$. 

For every Cartesian left additive category, there are always at least two
possible infinitesimal extensions: 

\begin{lemma} For any Cartesian left additive category $\mathbb{X}$, 
 \begin{enumerate}
 \item Setting $\varepsilon(f) = 0$ defines an infinitesimal extension on $\mathbb{X}$ and therefore in this case, $\oplus_A = \pi_0$ and $f \oplus g = f$.
\item Setting $\varepsilon(f) = f$ defines an infinitesimal extension on $\mathbb{X}$ and therefore in this case, $\oplus_A = +_A$ and $f \oplus g = f + g$.
\end{enumerate}
\end{lemma}
\begin{proof} This is straightforward and we leave it as an exercise for the reader. \hfill $\blacksquare$
\end{proof} 

We note that while these examples of infinitesimal extensions may seem trivial,
they are both very important as they will give rise to key examples of Cartesian
difference categories. 

\subsection{Cartesian Difference Categories}\label{CdCsec}

\begin{definition} A \textbf{Cartesian difference category} is a Cartesian left
additive category with an infinitesimal extension $\varepsilon$ which is
equipped with a \textbf{difference combinator} $\dd$ of the form:
\[ \frac{f : A \to B}{\dd[f]: A \times A \to B} \]
 verifying the following coherence conditions: 
 \begin{enumerate}[{\bf [C$\dd$.1]}]
 \setcounter{enumi}{-1}
 \item $f \circ (x + \varepsilon(y)) = f \circ x + \varepsilon\left( \dd[f] \circ \langle x, y \rangle \right)$ 
\item $\dd[f+g] = \dd[f] + \dd[g]$, $\dd[0] = 0$ and $\dd[\varepsilon(f)] = \varepsilon(\dd[f])$ 
 \item $\dd[f] \circ \langle x, y + z \rangle = \dd[f] \circ \langle x, y \rangle + \dd[f] \circ \langle x + \varepsilon(y), z \rangle$ and $\dd[f] \circ \langle x, 0 \rangle = 0$
 \item $\dd[1_A] = \pi_1$ and $\dd[\pi_0] = \pi_1; \pi_0$ and $\dd[\pi_1] = \pi_1; \pi_0$
  \item $\dd[\langle f, g \rangle] = \langle \dd[f], \dd[g] \rangle$ and $\dd[!_A] = !_{A \times A}$ 
 \item $\dd[g \circ f] = \dd[g] \circ \langle f \circ \pi_0, \dd[f] \rangle$ 
\item $\dd\left[\dd[f] \right] \circ \left \langle \langle  x, y \rangle, \langle 0, z \rangle \right \rangle= \dd[f] \circ  \langle x + \varepsilon(y), z \rangle$
\item  $ \dd\left[\dd[f] \right] \circ \left \langle \langle x, y \rangle, \langle z, 0 \rangle \right \rangle=  \dd\left[\dd[f] \right] \circ \left \langle \langle x, z \rangle, \langle y, 0 \rangle \right \rangle$
\end{enumerate}
\end{definition}

Before giving some intuition on the axioms {\bf [C$\dd$.0]} to {\bf [C$\dd$.7]}, we first observe that one could have used change action notation to express {\bf [C$\dd$.0]}, {\bf [C$\dd$.2]}, and {\bf [C$\dd$.6]} which would then be written as:
 \begin{description}
 \item[{\bf [C$\dd$.0]}] $f \circ (x \oplus y) = (f \circ x) \oplus \left(\dd[f] \circ \langle x, y \rangle \right)$ 
 \item[{\bf [C$\dd$.2]}]  $\dd[f] \circ \langle x, y + z \rangle = \dd[f] \circ \langle x, y \rangle + \dd[f] \circ \langle x \oplus y, z \rangle$ and $\dd[f] \circ \langle x, 0 \rangle = 0$
 \item[{\bf [C$\dd$.6]}]  $\dd\left[\dd[f] \right] \circ \left \langle \langle  x, y \rangle, \langle 0, z \rangle \right \rangle= \dd[f] \circ  \langle x \oplus y, z \rangle$
\end{description}
And also, just like Cartesian differential categories, {\bf [C$\dd$.6]} and {\bf [C$\dd$.7]} have alternative equivalent expressions. 
\begin{lemma} In the presence of the other axioms, {\bf [C$\dd$.6]} and {\bf [C$\dd$.7]} are equivalent to:
 \begin{enumerate}[{\bf [C$\dd$.1.{a}]}]
 \setcounter{enumi}{5}
\item $\dd\left[\dd[f] \right] \circ \left \langle \langle  x, 0 \rangle, \langle 0, y \rangle \right \rangle = \dd[f] \circ  \langle x, y \rangle$
\item $\dd\left[\dd[f] \right] \circ \left \langle \langle x, y \rangle, \langle z, w \rangle \right \rangle= \dd\left[\dd[f] \right] \circ \left \langle \langle x, z \rangle, \langle y, w \rangle \right \rangle$
\end{enumerate}
\end{lemma}
\begin{proof} The proof is essentially the same as \cite[Proposition 4.2]{cockett2014differential}. \hfill $\blacksquare$
\end{proof} 

The keen eyed reader will notice that the axioms of a Cartesian difference category are very similar to the axioms of a Cartesian differential category. Indeed, {\bf [C$\dd$.1]}, {\bf [C$\dd$.3]}, {\bf [C$\dd$.4]}, {\bf [C$\dd$.5]}, and {\bf [C$\dd$.7]} are the same as their Cartesian differential category counterpart. The axioms which are different are {\bf [C$\dd$.2]} and {\bf [C$\dd$.6]} where the infinitesimal extension $\varepsilon$ is now included, and also there is the new extra axiom {\bf [C$\dd$.0]}. On the other hand, interestingly enough, {\bf [C$\dd$.6.a]} is the same as {\bf [CD.6.a]}. We also point out that writing out {\bf [C$\dd$.0]} and {\bf [C$\dd$.2]} using change action notion, we see that these axioms are precisely {\bf [CAD.1]} and {\bf [CAD.2]} respectively. To better understand {\bf [C$\dd$.0]} to {\bf [C$\dd$.7]} it may be useful to write them out using element-like notation. 

In element-like notation, {\bf [C$\dd$.0]} is written as: 
\[
    f(x + \varepsilon(y)) = f(x) + \varepsilon\left(\dd[f](x,y) \right)
\]
This condition can be read as a generalization of the Kock-Lawvere axiom that characterizes the derivative in from synthetic differential geometry \cite{kock2006synthetic}. Broadly speaking, the Kock-Lawvere axiom states that, for any map $f : \mathcal{R} \to \mathcal{R}$ and any $x \in \mathcal{R}$ and $d \in \mathcal{D}$, there exists a unique $f'(x) \in \mathcal{R}$ verifying
\[
    f(x + d) = f(x) + d \cdot f'(x)
\]
where $\mathcal{D}$ is the subset of $\mathcal{R}$ consisting of infinitesimal elements. It is by analogy with the Kock-Lawvere axiom that we refer to $\varepsilon$ as an ``infinitesimal extension'' as it can be thought of as embedding the space $A$ into a subspace $\varepsilon(A)$ of infinitesimal elements. 

{\bf [C$\dd$.1]} states that the differential of a sum of maps is the sum of differentials, and similarly for zero maps and the infinitesimal extension of a map. 
%-- as we shall see later, it can also be thought of as stating that the maps $+_A, 0_A$ and $\varepsilon_A$ are linear. 
{\bf [C$\dd$.2]} is the first crucial difference between a Cartesian difference category and a Cartesian differential category. In a Cartesian differential category, the differential of a map is assumed to be additive in its second argument. In a Cartesian difference category, just as derivatives for change actions, while the differential is still required to preserve zeros in its second argument, it is only additive ``up to a small perturbation'', that is:
\[\dd[f](x,y+z) = \dd[f](x,y) + \dd[f](x + \varepsilon(y), z)\]
{\bf [C$\dd$.3]} tells us what the differential of the identity and projection maps are, while {\bf [C$\dd$.4]} says that the differential of a pairing of maps is the pairing of their differentials. {\bf [C$\dd$.5]} is the chain rule which expresses what the differential of a composition of maps is:
\[ \dd[g \circ f](x,y) = \dd[g](f(x), \dd[f](x,y))\]
{\bf [C$\dd$.6]} and {\bf [C$\dd$.7]} tell us how to work with second order differentials. {\bf [C$\dd$.6]} is expressed as follows: 
\[\dd \left[ \dd[f] \right]\left( x,y, 0,z \right) =  \dd[f](x + \varepsilon(y),
z) \]
and finally {\bf [C$\dd$.7]} is expressed as:
\[\dd \left[ \dd[f] \right]\left( x,y, z, 0 \right) =  \dd \left[ \dd[f] \right]\left( x, z, y, 0 \right) \]
It is interesting to note that while {\bf [C$\dd$.6]} is different from {\bf [CD.6]}, its alternative version {\bf [C$\dd$.6.a]} is the same as {\bf [CD.6.a]}. 
\[\dd \left[ \dd[f] \right]\left( (x,0), (0,y) \right) =  \dd[f](x, z) \]

Axiom {\bf [C$\dd$.6]} has some interesting and counter-intuitive consequences
which shall be of use later.

\newcommand\Item[1][]{%
  \ifx\relax#1\relax  \item \else \item[#1] \fi
  \abovedisplayskip=0pt\abovedisplayshortskip=0pt~\vspace*{-\baselineskip}
}

\begin{lemma}
  \label{lem:d-epsilon}
  Given any map $f : A \to B$ in a Cartesian difference category its
  derivatives satisfy the following equations for any $x, u, v, w : C \to A$:

  \begin{enumerate}[i.]
    \item $\dd[f] \circ \pair{x}{\varepsilon(u)} = \varepsilon(\dd[f]) \circ
    \pair{x}{u}$
    \label{lem:d-varepsilon-i}
    \item
    $\dd[f] \circ \pair{x}{u + v}
    = \dd[f] \circ \pair{x}{u} + \dd[f] \circ \pair{x + \varepsilon^2(u)}{v}$
    \label{lem:d-varepsilon-ii}
    \item $\varepsilon(\dd^2[f]) \circ \four{x}{u}{v}{0}
    = \varepsilon^2(\dd^2[f]) \circ \four{x}{u}{v}{0}$
    \label{lem:d-varepsilon-iii}
  \end{enumerate}
\end{lemma}
\begin{proof}\hspace{0pt}
  \begin{enumerate}[i.]
    \Item
    \begin{align*}
      \dd[f] \circ \pair{g}{\varepsilon(h)}
      = \dd[f]\circ \pair{g}{0}
      + \varepsilon(\dd^2[f]) \circ \four{g}{0}{0}{h}
      = \varepsilon(\dd[f]) \circ \pair{g}{h}
    \end{align*}
    \Item \begin{align*}
      &\varepsilon^2 (\dd^2[f]) \circ \four{x}{u}{v}{0}
      \\
      &= \varepsilon (\dd^2[f]) \circ \four{x}{\varepsilon(v)}{u}{0}
      &\text{(by \CdCax{7} and i.)}
      \\
      &= \varepsilon (\dd^3[f]) \circ \pair{
        \four{x}{0}{0}{v}
      }{
        \four{0}{0}{u}{0}
      }
      &\text{(by \CdCax{6})}
      \\
      &= \dd^3[f] \circ \pair{
        \four{x}{0}{0}{0}
      }{
        \four{0}{\varepsilon(v)}{\varepsilon(u)}{0}
      }
      &\text{(by \CdCax{7} and i.)}
      \\
      &= \dd^3[f] \circ \pair{
        \four{x}{0}{0}{\varepsilon(v)}
      }{
        \four{0}{0}{\varepsilon(u)}{0}
      }
      &\text{(by \CdCax{7})}
      \\
      &= \dd^2[f] \circ \four{x}{\varepsilon^2(v)}{\varepsilon(u)}{0}
      &\text{(by \CdCax{6})}
      \\
      &= \varepsilon^3(\dd^2[f]) \circ \four{x}{v}{u}{0}
      &\text{(by \CdCax{7} and i.)}
    \end{align*}
  \Item
    \begin{align*}
      &\varepsilon(\dd^2[f]) \circ \four{x}{u}{v}{0}
      \\
      &= \dd^2[f] \circ \four{x}{\varepsilon(u)}{v}{0}
      &\text{(by \CdCax{7} and i.)}
      \\
      &= \dd^3[f] \circ \pair{
        \four{x}{0}{0}{u}
      }{
        \four{0}{0}{v}{0}
      }
      &\text{(by \CdCax{6})}
      \\
      &= \dd^3[f] \circ \pair{
        \four{x}{0}{0}{0}
      }{
        \four{0}{u}{v}{0}
      }
      &\text{(by \CdCax{7})}
      \\
      &= \dd^3[f] \circ \pair{
        \four{x}{0}{0}{0}
      }{
        \four{0}{u}{0}{0}
      }\\
      &\quad + \dd^3[f] \circ \pair{
        \four{x}{\varepsilon^2(u)}{0}{0}
      }{
        \four{0}{0}{v}{0}
      }
      &\text{(by ii.)}
      \\
      &= \dd^3[f] \circ \pair{
        \four{x}{\varepsilon^2(u)}{0}{0}
      }{
        \four{0}{0}{v}{0}
      }
      &\text{(by \CdCax{7} and \CdCax{2})}
      \\
      &= \dd^2[f] \circ \four{x}{\varepsilon^2(u)}{v}{0}
      &\text{(by \CdCax{6})}
      \\
      &= \varepsilon^2(\dd^2[f]) \circ \four{x}{u}{v}{0}
      &\text{(by i.)}
    \end{align*}
  \end{enumerate}
  \hfill$\blacksquare$
\end{proof}

\begin{corollary}
  Let $\mathbb{X}$ be a Cartesian difference category with a nilpotent
  infinitesimal extension, that is, for any map $f : A \to B$ there is some
  $k \in \mathbb{N}$ such that $\varepsilon^k(f) = 0$. Then every
  derivative $\dd[f]$ is additive in its second argument.
\end{corollary}
\begin{proof}
  Fix a map $f : A \to B$, and consider arbitrary $x, u, v : C \to A$.
  By \CdCax{2}, it follows that $\dd[f]$ is additive ``up to a
  second-order term'', and so:
  \begin{align*}
    &\dd[f](x, u + v)\\ 
    &= \dd[f](x, u) + \dd[f](x + \varepsilon(u), v)
    &\text{(by \CdCax{2})}
    \\
    &= \dd[f](x, u) + \dd[f](x, v) + \varepsilon \dd^2[f](x, v, u, 0)
    &\text{(by \CdCax{0})}
    \\
    &= \dd[f](x, u) + \dd[f](x, v) + \varepsilon^k \dd^2[f](x, v, u, 0)
    &\text{(iterating Lemma \ref{lem:d-epsilon}.iii.)}
    \\
    &= \dd[f](x, u) + \dd[f](x, v)
    &\text{(since $\varepsilon$ is nilpotent)}
  \end{align*}
  \hfill$\blacksquare$
\end{proof}

Examples of Cartesian difference categories can be found in Section \ref{EXsec}.

\subsection{Another look at Cartesian Differential Categories}\label{CDCisCdCsec}

Here we explain how a Cartesian differential category is a Cartesian difference category where the infinitesimal extension is given by zero. 

\begin{proposition} Every Cartesian differential category $\mathbb{X}$ with differential combinator $\mathsf{D}$ is a Cartesian difference category where the infinitesimal extension is defined as $\varepsilon(f) = 0$ and the difference combinator is defined to be the differential combinator, $\dd = \mathsf{D}$.  
\end{proposition}
\begin{proof} As noted before, the first two parts of the {\bf [C$\dd$.1]}, the second part of {\bf [C$\dd$.2]}, {\bf [C$\dd$.3]}, {\bf [C$\dd$.4]}, {\bf [C$\dd$.5]}, and {\bf [C$\dd$.7]} are precisely the same as their Cartesian differential axiom counterparts. On the other hand, since $\varepsilon(f) =0$, {\bf [C$\dd$.0]} and the third part of {\bf [C$\dd$.1]} trivial state that $0=0$, while the first part of {\bf [C$\dd$.2]} and {\bf [C$\dd$.6]} end up being precisely the first part of {\bf [CD.2]} and {\bf [CD.6]}. Therefore, the differential combinator satisfies the Cartesian difference axioms and we conclude that a Cartesian differential category is a Cartesian difference category. 
\hfill $\blacksquare$
\end{proof} 

Conversely, one can always build a Cartesian differential category from a Cartesian difference category by considering the objects for which the infinitesimal extension is the zero map. 

\begin{proposition}\label{CdtoCD} For a Cartesian difference category $\mathbb{X}$ with infinitesimal extension $\varepsilon$ and difference combinator $\dd$, then $\mathbb{X}_{0}$, the full subcategory of objects $A$ such that $\varepsilon(1_A)= 0$, is a Cartesian differential category where the differential combinator is defined to be the difference combinator, $\mathsf{D} = \dd$. 
\end{proposition}
\begin{proof} First note that if $\varepsilon(1_A)=0$ and $\varepsilon(1_B)=0$, then by definition it also follows that $\varepsilon(1_{A \times B}) = 0$, and also that for the terminal object $\varepsilon(1_\top) = 0$ by uniqueness of maps into the terminal object. Thus $\mathbb{X}_{0}$ is closed under finite products and is therefore a Cartesian left additive category. Furthermore, we again note that since $\varepsilon(f) =0$, this implies that for maps between such objects the Cartesian difference axioms are precisely the Cartesian differential axioms. Therefore, the difference combinator is a differential combinator for this subcategory, and so $\mathbb{X}_{0}$ is a Cartesian differential category. 
\hfill $\blacksquare$
\end{proof} 

In any Cartesian difference category $\mathbb{X}$, the terminal object $\top$ always satisfies that $\varepsilon(1_\top) = 0$, and so therefore, $\mathbb{X}_{0}$ is never empty. On the other hand, applying Proposition \ref{CdtoCD} to a Cartesian differential category results in the entire category. It is also important to note that the above two propositions do not imply that if a difference combinator is a differential combinator then infinitesimal extension must be zero. In Section \ref{moduleex}, we provide such an example of a Cartesian differential category which comes equipped with a non-zero infinitesimal extension such that the differential combinator is a difference combinator with respect to this non-zero infinitesimal extension. 

\subsection{Cartesian Difference Categories as Change Action Models}\label{CdCisCAsec}

In this section we show how every Cartesian difference category is a particularly well-behaved change action model, and conversely how every change action model contains a Cartesian difference category. 

\begin{proposition} Let $\mathbb{X}$ be a Cartesian difference category with infinitesimal extension $\varepsilon$ and difference combinator $\dd$. Define  the functor $\alpha: \mathbb{X} \to \mathsf{CAct}(\mathbb{X})$ as $\alpha(A) = (A, A, \oplus_A, +_A, 0_A)$ (as defined in Lemma \ref{caelem}) and $\alpha(f) = (f, \dd[f])$. Then $(\mathbb{X}, \alpha: \mathbb{X} \to \mathsf{CAct}(\mathbb{X}))$ is a change action model.  
\end{proposition}
\begin{proof} By Lemma \ref{caelem}, $(A, A, \oplus_A, +_A, 0_A)$ is a change action and so $\alpha$ is well-defined on objects. While for a map $f$, $\dd[f]$ is a derivative of $f$ in the change action sense since {\bf [C$\dd$.0]} and {\bf [C$\dd$.2]} are precisely {\bf [CAD.1]} and {\bf [CAD.2]}, and so $\alpha$ is well-defined on maps. That $\alpha$ preserves identities and composition follows from {\bf [C$\dd$.3]} and {\bf [C$\dd$.5]} respectively, and so $\alpha$ is a functor. That $\alpha$ preserves finite products will follow from  {\bf [C$\dd$.3]} and {\bf [C$\dd$.4]}. Lastly, it is clear that $\alpha$ section of the forgetful functor, and therefore we conclude that $(\mathbb{X}, \alpha)$ is a change action model. 
\hfill $\blacksquare$
\end{proof} 

It is clear that not every change action model is a Cartesian difference category, since for example, change action models do not require addition to be commutative. On the other hand, it can be shown that every change action model contains a Cartesian difference category as a full subcategory.

\begin{definition}
    Let $(\mathbb{X}, \alpha : \mathbb{X} \to \mathsf{CAct}(\mathbb{X}))$ be a change action model. An object $A$ is \textbf{flat} whenever the following hold:
    \begin{enumerate}[{\bf [F.1]}]
    \item $\Delta A = A$
    \item $\alpha(\oplus_A) = (\oplus_A, \oplus_A \circ \pi_1)$
    \item $0 \oplus_A (0 \oplus_A f) = 0 \oplus_A f$ for any $f : U \to A$.
    \item $\oplus_A$ is right-injective, that is, if $\oplus_A \circ \langle f, g \rangle = \oplus_A \circ \langle f, h \rangle$ then $g=h$. 
    \end{enumerate}
\end{definition}

We would like to show that for any change action model $(\mathbb{X}, \alpha)$, its full subcategory of flat objects, $\mathsf{Flat}_\alpha$, is a Cartesian difference category. Starting with the finite product structure, since $\alpha$ preserves finite products, it is straightforward to see that $\top$ is flat and if $A$ and $B$ are flat then so is $A \times B$. The sum of maps $f: A \to B$ and $g: A \to B$ in $\mathsf{Flat}_\alpha$ is defined using the change action structure $f +_B g$, while the zero map $0: A \to B$ is $0 = 0_B \circ !_A$. And so we obtain that: 

\begin{lemma}\label{EUCLCLAC} $\mathsf{Flat}_\alpha$ is a Cartesian left additive category.
\end{lemma}
\begin{proof} Most of the Cartesian left additive structure is straightforward. However, since addition is not required to be commutative for arbitrary change actions, we will show that the addition is commutative for flat objects. Using that $\oplus_B$ is an action, that by {\bf [F.2]} we have that $\oplus_B \circ \pi_1$ is a derivative for $\oplus_B$, and {\bf [CAD.1]}, we obtain that: 
\begin{align*}
        0_B \oplus_B (f +_B g) = (0_B \oplus_B f) \oplus_B g
        = (0_B \oplus_B g) \oplus_B f
        = 0_B \oplus_B (g +_B f)
    \end{align*}
    By {[\bf F.4]}, $\oplus_B$ is right-injective and we conclude that $f + g = g + f$. \hfill $\blacksquare$
\end{proof}

We use the action of the change action structure to define the infinitesimal extension. So for a map $f: A \to B$ in $\mathsf{Flat}_\alpha$, define $\varepsilon(f): A \to B$ as follows:
\[ \varepsilon(f) = \oplus_B \circ \pair{0_B \circ{} !_A}{f} = 0 \oplus_B f \]

\begin{lemma} $\varepsilon$ is an infinitesimal extension for $\mathsf{Flat}_\alpha$. 
\end{lemma}
\begin{proof} We show that $\varepsilon$ preserves the addition. Following the
same idea as in the proof of Lemma \ref{EUCLCLAC}, we obtain the following: 
    \begin{align*}
 0_B \oplus_B \varepsilon(f +_B g) &= 0_B \oplus_B (0_B \oplus_B (f +_B g))\\
        &= (0_B \oplus_B 0_B) \oplus_B ((0_B \oplus_B f) \oplus_B g)\\
        &= (0_B \oplus_B (0_B \oplus_B f)) \oplus_B (0_B \oplus_B g)\\
        &= (0_B \oplus_B \varepsilon(f)) \oplus_B \varepsilon(g)\\
        &= 0_B \oplus_B (\varepsilon(f) +_B \varepsilon(g))
    \end{align*}
      Then by {[\bf F.3]}, it follows that $\varepsilon(f + g) = \varepsilon(f) + \varepsilon(g)$. The remaining infinitesimal extension axioms are proven in a similar fashion. \hfill $\blacksquare$
\end{proof}

Lastly, the difference combinator for $\mathsf{Flat}_\alpha$ is defined in the
obvious way, that is, $\dd[f]$ is defined as the second component of
$\alpha(f)$. 

\newcommand{\Fax}[1]{[{\bf F.{#1}}]}
\begin{proposition} Let $(\mathbb{X}, \alpha : \mathbb{X} \to
\mathsf{CAct}(\mathbb{X}))$ be a change action model. Then
$\mathsf{Flat}_\alpha$ is a Cartesian difference category.
\end{proposition}
\begin{proof}
  \CdCax{0} and \CdCax{2} are simply a restatement of the derivative condition.
  \CdCax{3} and \CdCax{4} follow immediately from the fact that $\alpha$
  preserves finite products and from the structure of products in
  $\mathsf{CAct}(\mathbb{X})$ (as per \cite[Section 3.2]{alvarez2019change}),
  while \CdCax{5} follows from the definition of composition in
  $\mathsf{CAct}(\mathbb{X})$.
  
  We now prove \CdCax{1}. First, by definition of $+$ and applying the chain
  rule, we have:
  \begin{align*}
    \dd[f + g] = \dd[+ \circ \pair{f}{g}] = \dd[+_{B}] \circ \pair{\pair{f}{g}
  \circ \pi_1}{ \pair{\dd[f]}{\dd[g]}}
  \end{align*}
  It suffices to show that $\dd[+] = + \circ \pi_1$, from which \CdCax{1}
  will follow. Consider arbitrary maps $u, w : A \to B$. Then (again using
  set-like notation) we have: 
  \begin{align*}
    0 \oplus \big(u \oplus w\big)
    = (0 \oplus u) \oplus (0 \oplus w)
    = 0 \oplus \big( u + (0 \oplus w)\big)
  \end{align*}
  By \Fax{4} we obtain the identity below:
  \begin{align}
    u \oplus w = u + \varepsilon(w)
    \label{oplus-epsilon} 
  \end{align}
  From this, it
  follows that, for any $u, v, w, l : A \to B$, the maps
  $ (u \oplus w) + (v \oplus l) $
  and
  $ u + v + \varepsilon (w) + \varepsilon(l) $ are identical. But, since $\varepsilon$
  is an infinitesimal extension, the second map can also be written as $
  (u + v) + \varepsilon (w + l)$, which is precisely $(u + v) \oplus (w + l)$.
  So we have
  \begin{align*}
     (u + v) \oplus (w + l) = (u \oplus w) + (v \oplus l) = (u + v) \oplus \dd[+](u,
  v, w, l)
  \end{align*}
  Applying \Fax{4} again gives $\dd[+] = + \circ \pi_1$ as desired. Axiom
  \CdCax{7} can be established in a similar way and we omit the calculations,
  proceeding directly to axiom \CdCax{6a} -- which is equivalent to, and easier
  to prove than, the more general \CdCax{6}. As before, we pick arbitrary $x, u
  : A \to B$ and calculate:
  \begin{align*}
    &f(x) \oplus \dd^2[f](x, 0, 0, u)
    \\
    &= f(x) + (0 \oplus \dd^2[f](x, 0, 0, u))
    &\text{(by (\ref{oplus-epsilon}))}
    \\
    &= f(x) + (0 \oplus (0 \oplus \dd^2[f](x, 0, 0, u))))
    &\text{(by \Fax{3})}
    \\
    &= f(x) + (0 \oplus (\dd[f](x, 0) \oplus \dd^2[f](x, 0, 0, u)))
    &\text{(by regularity)}
    \\
    &= f(x) + (0 \oplus \dd[f](x, 0 \oplus u))
    &\text{(by the derivative condition)}
    \\
    &= f(x) \oplus \dd[f](x, 0 \oplus u)
    &\text{(by \ref{oplus-epsilon})}
    \\
    &= f(x \oplus (0 \oplus u))
    &\text{(by the derivative condition)}
    \\
    &= f(x + 0 \oplus (0 \oplus u))
    &\text{(by \ref{oplus-epsilon})}
    \\
    &= f(x + 0 \oplus u)
    &\text{(by \Fax{3})}
    \\
    &= f(x \oplus u)
    &\text{(by \ref{oplus-epsilon})}
    \\
    &= f(x) \oplus \dd[f](x, u)
    &\text{(by the derivative condition)}
  \end{align*}
  Hence $\dd^2[f](x, 0, 0, u) = \dd[f](x, u)$ as desired.
  \hfill$\blacksquare $
\end{proof}

Note that the proof above goes through without appealing to \Fax{3}, except for
axiom \CdCax{6}, for which only a weaker version is provable. More precisely,
without \Fax{3}, one can only show
\[
  \varepsilon (\dd^2[f](x, 0, 0, u)) = \dd[f](x, \varepsilon(u))
\]
which is strictly weaker than \CdCax{6} and, in fact, holds trivially whenever
$\varepsilon = 0$.

\subsection{Linear Maps and $\varepsilon$-Linear Maps}

An important subclass of maps in a Cartesian differential category is the class
of \emph{linear maps} \cite[Definition 2.2.1]{blute2009cartesian}. One can also
define linear maps in a Cartesian difference category by using the same
definition. 

\begin{definition} \label{def:linearity} In a Cartesian difference category, a
map $f$ is \textbf{linear} if the following equality holds: $\dd[f] = f \circ
\pi_1$.
\end{definition}

Using element-like notation, a map $f$ is linear if $\dd[f](x,y) = f(y)$. Linear
maps in a Cartesian difference category satisfy many of the same properties
found in \cite[Lemma 2.2.2]{blute2009cartesian}.

\begin{lemma}\label{Lemma:linear} In a Cartesian difference category, 
\begin{enumerate}
    \item If $f: A \to B$ is linear then $\varepsilon(f) = f \circ \varepsilon(1_A)$; 
    \item If $f: A \to B$ is linear, then $f$ is additive (Definition \ref{LACdef});
    \item Identity maps, projection maps, and zero maps are linear; 
    \item The composite, sum, and pairing of linear maps is linear; 
    \item If $f: A \to B$ and $k: C \to D$ are linear, then for any map $g: B \to C$, the following equality holds: $\dd[k \circ g \circ f] = k \circ \dd[g] \circ (f \times f)$;
    \item If an isomorphism is linear, then its inverse is linear;
    \item For any object $A$, $\oplus_A$ and $+_A$ are linear. 
    \item Whenever $\varepsilon$ is nilpotent, then every derivative evaluated at
    zero is linear. That is, every map $\dd[f] \circ \pair{0}{1_A}$ is linear.
\end{enumerate}
\end{lemma}
\begin{proof}
  Most of these properties are straightforward consequences of the definition of
  linearity. We explicitly prove 2 and 8.

  For 2, suppose $f$ is linear. Then:
  \begin{align*}
    f(x + y) = \dd[f](0, x + y)
    = \dd[f](0, x) + \dd[f](0 + \varepsilon(x), y)
    = f(x) + f(y)
  \end{align*}

  For 8, suppose $\varepsilon^k (f) = 0$ for all $f$. Then:
  \begin{align*}
    \dd[\dd[f] \circ \pair{0}{1_A}]
    &= \dd[\dd[f]] \circ \pair{\pair{0}{\pi_0}}{\pair{0}{\pi_1}}
    &\text{(by \CdCax{})}
    \\
    &= \dd[f] \circ \pair{0 + \varepsilon(\pi_0)}{\pi_1}
    &\text{(by \CdCax{6})}
    \\
    &= \big(\dd[f] \circ \pair{0}{\pi_1}\big)\\
    &\quad + \varepsilon\big(\dd[\dd[f]] \circ \pair{\pair{0}{\pi_1}}{\pair{\pi_0}{0}}\big)
    &\text{(by \CdCax{0})}
    \\
    &= \big(\dd[f] \circ \pair{0}{\pi_1}\big)\\
    &\quad + \varepsilon^k\big(\dd[\dd[f]] \circ \pair{\pair{0}{\pi_1}}{\pair{\pi_0}{0}}\big)
    &\text{(iterating Lemma \ref{lem:d-epsilon}.iii.)}
    \\
    &= \big(\dd[f] \circ \pair{0}{\pi_1}\big)
    &\text{(since $\varepsilon$ is nilpotent)}
  \end{align*}
  All the other properties follow much the same proofs in 
  \cite[Lemma 2.2.2]{blute2009cartesian}.
  \hfill $\blacksquare$
\end{proof}

Using element-like notation, the first point of the above lemma says that if $f$ is linear then $f(\varepsilon(x)) = \varepsilon(f(x))$. And while all linear maps are additive, the converse is not necessarily true, see \cite[Corollary 2.3.4]{blute2009cartesian}. However, an immediate consequence of the above lemma is that the subcategory of linear maps of a Cartesian difference category has finite biproducts. %We also note a particular omission from \cite[Lemma 2.2.2]{blute2009cartesian} for an arbitrary Cartesian difference category due to \textbf{[C$\dd$.6]}, which is that $\dd[f] \circ \langle 0, 1_A \rangle$ is no longer necessarily a linear map.

Another interesting subclass of maps in a Cartesian difference category are the $\varepsilon$-linear maps, which are maps whose infinitesimal extension is linear. 

\begin{definition} In a Cartesian difference category, a map $f$ is \textbf{$\varepsilon$-linear} if $\varepsilon(f)$ is linear.
\end{definition}

\begin{lemma} In a Cartesian difference category, 
\begin{enumerate}
\item If $f: A \to B$ is $\varepsilon$-linear then $f \circ (x + \varepsilon(y)) = f \circ x + \varepsilon(f) \circ y$; 
    \item Every linear map is $\varepsilon$-linear;
    \item The composite, sum, and pairing of $\varepsilon$-linear maps is $\varepsilon$-linear; 
    \item If an isomorphism is $\varepsilon$-linear, then its inverse is again $\varepsilon$-linear. 
\end{enumerate}
\end{lemma}

Using element-like notation, the first point of the above lemma says that if $f$ is $\varepsilon$-linear then $f(x + \varepsilon(y)) = f(x) + \varepsilon(f(y))$. So $\varepsilon$-linear maps are additive on ``infinitesimal'' elements (i.e. those of the form $\varepsilon(y)$). 

For a Cartesian differential category, linear maps in the Cartesian difference category sense are precisely the same as Cartesian differential category sense \cite[Definition 2.2.1]{blute2009cartesian}, while every map is $\varepsilon$-linear since $\varepsilon =0$. 

\section{Examples of Cartesian Difference Categories}\label{EXsec}

\subsection{Smooth Functions}\label{smoothex}

Every Cartesian differential category is a Cartesian difference category where the infinitesimal extension is zero. As a particular example, we consider the category of real smooth functions, which as mentioned above, can be considered to be the canonical (and motivating) example of a Cartesian differential category.

Let $\mathbb{R}$ be the set of real numbers and let $\mathsf{SMOOTH}$ be the category whose objects are Euclidean spaces $\mathbb{R}^n$ (including the point $\mathbb{R}^0 = \lbrace \ast \rbrace$), and whose maps are smooth functions $F: \mathbb{R}^n \to \mathbb{R}^m$. $\mathsf{SMOOTH}$ is a Cartesian left additive category where the product structure is given by the standard Cartesian product of Euclidean spaces and where the additive structure is defined by point-wise addition, $(F+G)(\vec x) = F(\vec x) + G(\vec x)$ and $0(\vec x) = (0, \hdots, 0)$, where $\vec x \in \mathbb{R}^n$. $\mathsf{SMOOTH}$ is a Cartesian differential category where the differential combinator is defined by the directional derivative of smooth functions. Explicitly, for a smooth function $F: \mathbb{R}^n \to \mathbb{R}^m$, which is in fact a tuple of smooth functions $F= (f_1, \hdots, f_n)$ where $f_i: \mathbb{R}^n \to \mathbb{R}$, $\mathsf{D}[F]: \mathbb{R}^n \times \mathbb{R}^n \to \mathbb{R}^m$ is defined as follows: 
\[\mathsf{D}[F]\left(\vec x, \vec y \right) := \left( \sum \limits^n_{i=1} \frac{\partial f_1}{\partial u_i}(\vec x) y_i, \hdots, \sum \limits^n_{i=1} \frac{\partial f_n}{\partial u_i}(\vec x) y_i  \right)\]
where $\vec x = (x_1, \hdots, x_n), \vec y = (y_1, \hdots, y_n) \in \mathbb{R}^n$. Alternatively, $\mathsf{D}[F]$ can also be defined in terms of the Jacobian matrix of $F$. Therefore $\mathsf{SMOOTH}$ is a Cartesian difference category with infinitesimal extesion $\varepsilon =0$ and with difference combinator $\mathsf{D}$. Since $\varepsilon = 0$, the induced action is simply $\vec x \oplus_{\mathbb{R}^n} \vec y = \vec x$. Also a smooth function is linear in the Cartesian difference category sense precisely if it is $\mathbb{R}$-linear in the classical sense, and every smooth function is $\varepsilon$-linear. 

\subsection{Calculus of Finite Differences}\label{discreteex}

Here we explain how the difference operator from the calculus of finite differences gives an example of a Cartesian difference category but \emph{not} a Cartesian differential category. This example was the main motivating example for developing Cartesian difference categories. The calculus of finite differences is captured by the category of abelian groups and arbitrary set functions between them. 

Let $\overline{\mathsf{Ab}}$ be the category whose objects are abelian groups $G$ (where we use additive notation for group structure) and where a map $f: G \to H$ is simply an arbitrary function between them (and therefore does not necessarily preserve the group structure). $\overline{\mathsf{Ab}}$ is a Cartesian left additive category where the product structure is given by the standard Cartesian product of sets and where the additive structure is again given by point-wise addition, $(f+g)(x)=f(x) + g(x)$ and $0(x)=0$. $\overline{\mathsf{Ab}}$ is a Cartesian difference category where the infinitesimal extension is simply given by the identity, that is, $\varepsilon(f)=f$, and and where the difference combinator $\dd$ is defined as follows for a map $f: G \to H$: 
 \[\dd[f](x,y) = f(x + y) - f(x)\]
On the other hand, $\dd$ is not a differential combinator for $\overline{\mathsf{Ab}}$ since it does not satisfy {\bf [CD.6]} and part of {\bf [CD.2]}. 
%Indeed, since as $f$ is not necessarily a group homomorphism, $\dd[f]$ fails {\bf [CD.2]} as: 
%\[\dd[f](x,y +z) = f(x + y +z) - f(x)\]
%is not necessarily equal to: 
%\[\dd[f](x,y) + \dd[f](x,z)= f(x+y) - f(x) + f(x+z) - f(x)\] 
Thanks to the addition of the infinitesimal extension, $\dd$ does satisfy {\bf [C$\dd$.2]} and {\bf [C$\dd$.6]}, as well as {\bf [C$\dd$.0]}. 
%, which in this case are respectively: 
%\[\dd[f](x,y+z) = \dd[f](x,y) + \dd[f](x + y, z)\] 
%\[\dd \left[ \dd[f] \right]\left( (x,y), (0,z) \right) =  \dd[f](x +y, z) \]
%\[ f(x + y) = f(x) + \dd[f](x,y) \]
However, as noted in \cite{FMCS2018}, it is interesting to note that this $\dd$ does satisfy {\bf [CD.1]}, the second part of {\bf [CD.2]}, {\bf [CD.3]}, {\bf [CD.4]}, {\bf [CD.5]}, {\bf [CD.7]}, and {\bf [CD.6.a]}. It is worth noting that in \cite{FMCS2018}, the goal was to drop the addition and develop a ``non-additive'' version of Cartesian differential categories. 

In $\overline{\mathsf{Ab}}$, since the infinitesimal operator is given by the identity, the induced action is simply the addition, $x \oplus_G y = x + y$. On the other hand, the linear maps in $\overline{\mathsf{Ab}}$ are precisely the group homomorphisms. Indeed, $f$ is linear if $\dd[f](x,y) = f(y)$. But by {\bf [C$\dd$.0]} and {\bf [C$\dd$.2]}, we get that: 
\[f(x + y) = f(x) + \dd[f](x,y)= f(x) + f(y) \quad \quad \quad f(0) =  \dd[f](x,0) = 0 \]
So $f$ is a group homomorphism. Conversely, if $f$ is a group homomorphism: 
\[\dd[f](x,y) = f(x+y) - f(x) = f(x) + f(y) - f(x) = f(y)\]
So $f$ is linear. Since $\varepsilon(f)=f$, the $\varepsilon$-linear maps are precisely the linear maps. 

\subsection{Module Morphisms}\label{moduleex}

Here we provide a simple example of a Cartesian difference category whose difference combinator is also a differential combinator, but where the infinitesimal extension is neither zero nor the identity. 

Let $R$ be a commutative semiring and let $\mathsf{MOD}_R$ be the category of $R$-modules and $R$-linear maps between them. $\mathsf{MOD}_R$ has finite biproducts and is therefore a Cartesian left additive category where every map is additive. Every $r \in R$ induces an infinitesimal extension $\varepsilon^r$ defined by scalar multiplication, $\varepsilon^r(f)(m) = r f(m)$. Then $\mathsf{MOD}_R$ is a Cartesian difference category with the infinitesimal extension $\varepsilon^r$ for any $r \in R$ and difference combinator $\dd$ defined as:
\[\dd[f](m,n)=f(n)\]
$R$-linearity of $f$ assures that {\bf [C$\dd$.0]} holds, while the remaining Cartesian difference axioms hold trivially. In fact, $\dd$ is also a differential combinator and therefore $\mathsf{MOD}_R$ is also a Cartesian differential category. The induced action is given by $m \oplus_M n = m + rn$. By definition of $\dd$, every map in $\mathsf{MOD}_R$ is linear, and by definition of $\varepsilon^r$ and $R$-linearity, every map is also $\varepsilon$-linear.  

\newcommand{\seq}[1]{\left[ {#1} \right]}
\newcommand{\Ab}[0]{\overline{\mathsf{Ab}}}
\newcommand{\z}[0]{\mathbf{z}}
\subsection{Stream calculus}\label{streamex}
Here we show how one can extend the calculus of finite differences example to stream calculus. 

For a set $A$, let $A^\omega$ denote the set of infinite sequences of elements of $A$, where we write $\seq{a_i}$ for the infinite sequence $\seq{a_i} = (a_1, a_2, a_3, \hdots)$ and $a_{i:j}$ for the (finite) subsequence $(a_i, a_{i + 1}, \hdots, a_j)$. A function $f : A^\omega \to B^\omega$ is \textbf{causal} whenever the
  $n$-th element $f\left(\seq{a_i}\right)_n$ of the output sequence only depends on the first $n$ elements of $\seq{a_i}$, that is, $f$ is causal if and only if whenever $a_{0:n} = b_{0:n}$ 
  then $f\left(\seq{a_i}\right)_{0:n} = f\left(\seq{b_i}\right)_{0:n}$. We now consider streams over abelian groups, so let $\Ab^\omega$ be the category whose objects are all the Abelian groups and whose morphisms are causal maps from $G^\omega$ to $H^\omega$. $\Ab^\omega$ is a Cartesian left-additive category, where the product is given by the standard product of abelian groups and where the additive structure is lifted point-wise from the structure of $\Ab$, that is, $(f+g)\left(\seq{a_i}\right)_n = f\left(\seq{a_i}\right)_n + g\left(\seq{a_i}\right)_n$ and $0\left(\seq{a_i}\right)_n = 0$. In order to define the infinitesimal extension, we first need to define the truncation operator $\z$. So let $G$ be an abelian group and $\seq{a_i} \in G^\omega$, then define the sequence $\z (\seq{a_i})$ as:
    \begin{align*}
   \z (\seq{a_i})_0 = 0  && 
\z \left(\seq{a_i}\right)_{n+1} = a_{n+1}
  \end{align*}
   The category $\Ab^\omega$ is a Cartesian difference category where the infinitesimal extension is given by the truncation operator, $\varepsilon(f)\left(\seq{a_i}\right) = \z\left( f\left(\seq{a_i}\right) \right)$, 
  %  is defined as follows: 
 % \begin{align*}
  %  \varepsilon(f)\left(\seq{a_i}\right)_0 = 0  && 
  %\varepsilon(f)\left(\seq{a_i}\right)_{n+1} = f\left(\seq{a_i}\right)_n
 % \end{align*}
  and where the difference combinator $\dd$ is defined as follows: 
    \begin{align*}
           \dd[f]\left(\seq{a_i}, \seq{b_i}\right)_0 &= f\left(\seq{a_i} + \seq{b_i}\right)_0 -  f\left(\seq{a_i}\right)_0 \\
                \dd[f]\left(\seq{a_i}, \seq{b_i}\right)_{n + 1} &= f\left(\seq{a_i} + \z(\seq{b_i}) \right)_{n + 1} -  f\left(\seq{a_i}\right)_{n+1}
    \end{align*}
Note the similarities between the difference combinator on $\Ab$ and that on $\Ab^\omega$. The induced action is computed out to be: 
  \begin{align*}
   ( \seq{a_i} \oplus \seq{b_i})_0 = a_0 && ( \seq{a_i} \oplus \seq{b_i})_{n+1} = a_{n+1} + b_{n+1}   \end{align*}
A causal map is linear (in the Cartesian difference category sense) if and only if it is a group homomorphism. While a causal map $f$ is $\varepsilon$-linear if and only if it is a group homomorphism which does not the depend on the $0$-th term of its input, that is, $f\left(\seq{a_i}\right) = f\left( \z (\seq{a_i}) \right)$. 

\iffalse
\begin{definition}
  Given an Abelian group $A$, we define the \emph{truncation operator} $\z_A : A^\omega \to A^\omega$ by
 % \begin{align*}
    %(\z \seq{a_i})_0 &= 0\\
    %(\z \seq{a_i})_{j + 1} &= a_{j + 1}
  %\end{align*}
  $\z_A$ is a monoid homomorphism according to the pointwise monoid structure on $A^\omega$. Thus all the $\z_A$ define an infinitesimal extension $\z$ for $\Ab^\omega$.
\end{definition}

\begin{theorem}
    The category $\Ab^\omega$ is a Cartesian difference category, with difference operator given by:
    \begin{align*}
        \dd[f](x, u)_0 = f(x + u)_0 - f(x)_0
        \dd[f](x, u)_{i + 1} = f(x + \z u)_{i + 1} - f(x)_{i + 1}
    \end{align*}
\end{theorem}
\begin{proof}\red{only for my own reference, this can/should be deleted!}
    The derivative condition is satisfied trivially because of causality.

    If $f = Id$ we have $\dd[Id](x, u)_0 = u_0$ and $\dd[Id](x, u)_{i + 1} = u_{i + 1}$ for all $i$, hence $\dd[Id] = Id \circ \pi_1$ (and similarly for $\pi_0, \pi_1, 0, +$)
    
    If $f = \varepsilon = \z$ then $\dd[\z](x, u)_0 = \z (x + u)_0 - \z(x)_0 = 0$ and $\dd[\z](x, u)_{i + 1} = \z (x + u)_{i + 1} - \z(x)_{i + 1} = u_{i + 1}$ so $\dd[\z](x, u) = \z(u)$ as expected.
    
    Axioms 4-7 hold quite easily because these derivatives are just the "lifting" of the group derivative.
\end{proof}
\fi

\newcommand{\T}[0]{\mathsf{T}}
\section{Tangent Bundles in Cartesian Difference Categories} \label{monadsec}

In this section we show that the difference combinator of a Cartesian difference category induces a monad, called the \emph{tangent monad}, whose Kleisli category is again a Cartesian difference category. This construction is a generalization of the tangent monad for Cartesian differential categories \cite{cockett2014differential,manzyuk2012tangent}. 
%However, the Kleisli category of the tangent monad of Cartesian differential category is \emph{not} a Cartesian differential category, but rather a Cartesian difference category. 

\subsection{The Tangent Bundle Monad}

Let $\mathbb{X}$ be a Cartesian difference category with infinitesimal extension $\varepsilon$ and difference combinator $\dd$. Define the functor $\mathsf{T}: \mathbb{X} \to \mathbb{X}$ as follows: 
\[ \mathsf{T}(A) = A \times A \quad \quad \quad \mathsf{T}(f) = \langle f \circ \pi_0, \dd[f] \rangle \]
and define the natural transformations $\eta: \mathsf{1}_\mathbb{X} \Rightarrow \mathsf{T}$ and $\mu: \mathsf{T}^2 \Rightarrow \mathsf{T}$ as follows: 
\[\eta_A := \langle 1_A, 0 \rangle \quad \quad \quad \mu_A := \left \langle \pi_0\circ \pi_0, \pi_1 \circ \pi_0 + \pi_0 \circ \pi_1 + \varepsilon(\pi_1 \circ \pi_1) \right \rangle \]

%\begin{lemma} $\mathsf{T}: \mathbb{X} \to \mathbb{X}$ is a functor, and $\eta$ and $\mu$ are natural transformations. 
%\end{lemma}
%\begin{proof} See Appendix. \hfill $\blacksquare$
%\end{proof} 

\begin{proposition} $(\mathsf{T}, \mu, \eta)$ is a monad. 
\end{proposition} 
\begin{proof}
  \renewcommand{\T}[0]{\mathsf{T}}
  \newcommand{\Id}[1]{1_{#1}}
  \newcommand{\pa}[1]{\left({#1}\right)}
  Functoriality of $\T$ follow immediately from \textbf{[C$\dd$.3]} and the
  chain rule \textbf{[C$\dd$.5]}. Naturality of $\eta$ and $\mu$ and the monad
  identities follow from the remaining difference combinator axioms. 
  It remains to check that
  the monad laws hold, which is a matter of simple calculation.
  \begin{align*}
    \mu_A \circ \eta_{\T(A)}
    &= \mu_A \circ \pair{\Id{\T(A)}}{0}\\
    &= \pair{\pi_{00}}{\pi_{10} + \pi_{01} + \varepsilon(\pi_{11})}
    \circ \pair{\Id{\T(A)}}{0}
    \\
    &= \pair{\pi_0}{\pi_1 + 0 + \varepsilon(0)}
    \\
    &= \Id{\T(A)}
  \end{align*}\vspace{-1.5\baselineskip}
  \begin{align*}
    \mu_A \circ \T(\eta_A)
    &= \mu_A \circ \four{\pi_0}{0}{\pi_1}{0}\\
    &= \pair{\pi_{00}}{\pi_{10} + \pi_{01} + \varepsilon(\pi_{11})}
    \circ \four{\pi_0}{0}{\pi_1}{0}
    \\
    &= \pair{\pi_0}{0 + \pi_1 + \varepsilon(0)}\\
    &= \Id{\T(A)}
  \end{align*}
  where $\pi_{ij} = \pi_i \circ \pi_j$. For the last of the monad laws, we first note that, since $\mu$ is linear,
  it follows that $\T(\mu) = \mu \times \mu$. Then it suffices to compute:
  \begin{align*}
    \mu_A \circ \T(\mu_A)
    &= \mu_A \circ (\mu_A \times \mu_A)
    \\
    &= \pair{\pi_{00}}{\pi_{10} + \pi_{01} + \varepsilon(\pi_{11})}
    \circ (\mu_A \times \mu_A)
    \\
    &=
    \pair{\pi_0 \circ \mu_A \circ \pi_1}{\pi_1 \circ \mu_A \circ \pi_1
    + \pi_0 \circ \mu_A \circ \pi_1 + \varepsilon(\pi_1 \circ \mu_A \circ \pi_1)}
    \\
    &=
    \pair{\pi_{000}}{
      \pi_{100} + \pi_{010} + \varepsilon(\pi_{001})
      + \pi_{001}
      + \varepsilon \pa{\pi_{010} + \pi_{011} + \varepsilon(\pi_{111})}
    }
    \\
    &= 
    \pair{\pi_{00}}{\pi_{10} + \pi_{01} + \varepsilon(\pi_{11}) }
    \circ \pair{\pi_{00}}{\pi_{10} + \pi_{01} + \varepsilon(\pi_{11}) }
    \\
    &= \mu_A \circ \mu_{\T(A)}
  \end{align*}
  So we conclude that $(T. \mu, \eta)$ is a monad.   \hfill $\blacksquare$
\end{proof} 

When $\mathbb{X}$ is a Cartesian differential category with the difference structure arising from setting $\varepsilon = 0$, this tangent bundle monad coincides with the standard tangent monad corresponding to its tangent category structure \cite{cockett2014differential,manzyuk2012tangent}.

\iffalse
\begin{example} \normalfont For a Cartesian differential category, since $\varepsilon = 0$, the induced monad is precisely the monad induced by its tangent category structure \cite{cockett2014differential,manzyuk2012tangent}. For example, in the Cartesian differential category $\mathsf{SMOOTH}$ (as defined in Section \ref{smoothex}), one has that $\mathsf{T}(F)(\vec x, \vec y) = (F(\vec x), \mathsf{D}[F](\vec x, \vec y)$, $\eta_{\mathbb{R}^n}(\vec x) = (\vec x,\vec 0)$, and $\mu_{\mathbb{R}^n}((\vec x,\vec y), (\vec z, \vec w)) = (\vec x, \vec y + \vec z)$. 
\end{example}

\begin{example} \normalfont In the Cartesian difference category $\overline{\mathsf{Ab}}$ (as defined in Section \ref{discreteex}), the monad is given by $\mathsf{T}(f)(x,y) = (f(x), f(x+y) - f(x))$, $\eta_G(x) = (x,0)$, and $\mu_G((x,y),(z,w)) = (x, y + z + w)$. 
\end{example}

\begin{example} \normalfont In the Cartesian difference category $\mathsf{MOD}_R$ (as defined in Section \ref{moduleex}) with infinitesimal extension $\varepsilon^r$, for $r \in R$, $\mathsf{T}(f)(m,n) = (f(m), f(n))$, $\eta_M(m) = (m,0)$, and $\mu_M((m,n),(p,q)) = (m, n + p + rq)$. 
\end{example}
\fi

\subsection{The Kleisli Category of $\mathsf{T}$}
In this section we show that the Kleisli category of the tangent monad is a Cartesian difference category\footnote{The construction found here is different from the one found in the conference paper \cite{alvarez2020cartesian}. Indeed, the proposed difference combinator in \cite{alvarez2020cartesian} was based on the one that appeared in \cite{alvarez2019change}. Unfortunately, we have found that said proposed difference combinator fails to satisfy {\bf [C$\dd$.2]} and therefore both of the aforementioned results in \cite{alvarez2020cartesian,alvarez2019change} are incorrect.}.
Recall that the Kleisli category of the monad $(\mathsf{T}, \mu, \eta)$ is defined as the category $\mathbb{X}_\mathsf{T}$ whose objects are the objects of $\mathbb{X}$, and where a map $A \to B$ in $\mathbb{X}_\mathsf{T}$ is a map $f: A \to \mathsf{T}(B)$ in $\mathbb{X}$, which would be a pair $f = \langle f_0, f_1 \rangle$ where $f_j: A \to B$. The identity map in $\mathbb{X}_\mathsf{T}$ is the monad unit $\eta_A: A \to \mathsf{T}(A)$, while composition of Kleisli maps $f: A \to \mathsf{T}(B)$ and $g: B \to \mathsf{T}(C)$ is defined as the composite $\mu_C \circ \mathsf{T}(g) \circ f$. To distinguish between composition in $\mathbb{X}$ and $\mathbb{X}_\mathsf{T}$, we denote the Kleisli composition as follows: 
\[ g \circ^\T f = \mu_C \circ \mathsf{T}(g) \circ f \]
If $f=\langle f_0, f_1 \rangle$ and $g=\langle g_0, g_1 \rangle$, then their Kleisli composition can be worked out to be: 
\[ g \circ^\T f = \langle g_0, g_1 \rangle \circ^\T \langle f_0, f_1 \rangle = \left \langle g_0 \circ f_0, \dd[g_0] \circ \langle f_0, f_1 \rangle + g_1 \circ (f_0 \oplus f_1) \right \rangle \]
Kleisli maps can be understood as ``generalized'' vector fields. Indeed, $\T(A)$ should be thought of as the tangent bundle over $A$, and therefore a vector field would be a map $\langle 1, f \rangle: A \to \T(A)$, which is of course also a Kleisli map. For more details on the intuition behind this Kleisli category see \cite{cockett2014differential}. 

We now wish to explain how the Kleisli category is again a Cartesian difference category. We begin by exhibiting the Cartesian left additive structure of the Kleisli category. First note that $\T(A \times B) \cong \T(A) \times \T(B)$ via the canonical natural isomorphism $\phi = \four{\pi_{00}}{\pi_{01}}{\pi_{10}}{\pi_{11}}$ (where recall $\pi_{ij} = \pi_i \circ \pi_j$). As such, the product of objects in $\mathbb{X}_\mathsf{T}$ is defined as $A \times B$ with projections $\pi^{\mathsf{T}}_0: A \times B \to \mathsf{T}(A)$ and $\pi^{\mathsf{T}}_1: A \times B \to \mathsf{T}(B)$ defined respectively as $\pi^{\mathsf{T}}_0 = \langle \pi_0, 0 \rangle$ and $\pi^{\mathsf{T}}_1 = \langle \pi_1, 0 \rangle$, and the pairing of Kleisli maps $f=\langle f_0, f_1 \rangle$ and $g=\langle g_0, g_1 \rangle$ is defined as:
\[ \langle f, g \rangle^\mathsf{T} = \phi \circ \pair{f}{g} = \four{f_0}{g_0}{f_1}{g_1} \]
The terminal object is again $\top$ and where the unique map to the terminal object is $!^{\mathsf{T}}_A = 0$. The sum of Kleisli maps $f=\langle f_0, f_1 \rangle$ and $g=\langle g_0, g_1 \rangle$ is defined as:
\[f +^\mathsf{T} g = f + g = \langle f_0 + g_0, f_1 + g_1 \rangle\]
and the zero Kleisli maps is simply $0^\T = 0 = \langle 0, 0 \rangle$. Therefore we conclude that the Kleisli category of the tangent monad is a Cartesian left additive category. 

\begin{lemma} $\mathbb{X}_\mathsf{T}$ is a Cartesian left additive category. 
\end{lemma} 

The infinitesimal extension $\varepsilon^\T$ for the Kleisli category is the same as the infinitesimal extension of the base category, that is, for a Kleisli map $f= \langle f_0, f_1 \rangle$: 
\[ \varepsilon^\T(f) = \varepsilon(f) = \pair{\varepsilon(f_0)}{\varepsilon(f_1)}\]
and so it follows that: 
\[ f \oplus^\T g = f +^\T \varepsilon^\T(g) = f + \varepsilon(g) = f \oplus g \]

\begin{lemma} $\varepsilon^\T$ is an infinitesimal extension on $\mathbb{X}_\mathsf{T}$. 
\end{lemma} 

To define the difference combinator for the Kleisli category, first note that difference combinators by definition do not change the codomain. That is, if $f : A \to \T(B)$ is a Kleisli arrow, then the type of its derivative \emph{qua} Kleisli arrow should be $A \times A \to \T(B) \times \T(B)$, which coincides with the type of its derivative in $\mathbb{X}$. Therefore, the difference combinator $\dd^\T$ for the Kleisli category can be defined to be the difference combinator of the base category, that is, for a Kleisli map $f= \langle f_0, f_1 \rangle$:
\[\dd^\T[f] = \dd[f] = \langle \dd[f_0], \dd[f_1] \rangle\] 

\begin{proposition} For a Cartesian difference category $\mathbb{X}$, the Kleisli category $\mathbb{X}_\mathsf{T}$ is a Cartesian difference category with infinitesimal extension $\varepsilon^\T$ and difference combinator $\dd^\T$.
\end{proposition}
\begin{proof}   
We first note that, for any map $f$ in $\mathbb{X}$, the following equality holds: 
  \[ 
    \T(\dd[f]) \circ \phi = 
    \dd[{\T(f)}] 
  \] 
This will help simplify many of the calculations to follow, since
  $\T(\dd[f])$ appears everywhere due to the definition of Kleisli
  composition. Indeed, note that we can first compute that: 
  \begin{align*}
\dd^T[f] \circ^\T \pair{x}{y}^\T &=~ \mu \circ \T(\dd^T[f]) \circ \pair{x}{y}^\T \\
&=~  \mu \circ \T(\dd[f]) \circ \phi \circ \pair{x}{y} \\
&=~ \mu \circ  \dd[{\T(f)}] \circ \pair{x}{y}
\end{align*}
So we have that $\dd^T[f] \circ^\T \pair{x}{y}^\T = \mu \circ  \dd[{\T(f)}] \circ \pair{x}{y}$. 
  
We now prove the seven Cartesian difference category axioms. \\\\
\noindent {\bf [C$\dd$.0]} $f \circ^\T (x +^\T \varepsilon^\T(y)) = f \circ^\T x +^\T \varepsilon^\T\left( \dd^\T[f] \circ^\T \langle x, y \rangle^\T \right)$ \\\\
    First note that $\mu$ is linear and therefore additive. Then we compute that: 
     \begin{align*}
f \circ^\T x +^\T \varepsilon^\T\left( \dd^\T[f] \circ^\T \langle x, y \rangle^\T \right) &=~ f \circ^\T x + \varepsilon\left( \dd^\T[f] \circ^\T \langle x, y \rangle^\T \right) \\
&=~ \mu \circ \T(f) \circ x + \varepsilon\left( \mu \circ  \dd[{\T(f)}] \circ \pair{x}{y} \right) \\
&=~ \mu \circ \T(f) \circ x + \mu \circ \varepsilon\left( \dd[{\T(f)}] \circ \pair{x}{y} \right) \tag{$\mu$ linear and Lem.\ref{Lemma:linear}.1} \\
&=~ \mu \circ \left(\T(f) \circ x + \varepsilon\left( \dd[{\T(f)}] \circ \pair{x}{y} \right) \right) \tag{$\mu$ is additive by Lem.\ref{Lemma:linear}.2} \\
&=~\mu \circ \T(f) \circ (x + \varepsilon(y))  \tag{by {\bf [C$\dd$.0]}} \\
&=~ f \circ^\T (x + \varepsilon(y)) \\
&=~f \circ^\T (x +^\T \varepsilon^\T(y))
\end{align*}
\noindent {\bf [C$\dd$.1]}  $\dd^\T[f +^\T g] = \dd^\T[f] +^\T \dd^\T[g]$, $[\dd^\T[0^\T] =  \dd[0]$, and $\dd^\T[\varepsilon^\T(f)] = \varepsilon^\T\left( \dd^\T[f] \right)$ \\\\
    Since both the sum, zero maps, infinitesimal extension, and differential combinator in the Kleisli category are the same as in the base category, by {\bf [C$\dd$.1]} it easily follows that:
    \begin{align*}
\dd^\T[f +^\T g] = \dd[f+g] = \dd[f] + \dd[g] = \dd^\T[f] +^\T \dd^\T[g] 
\end{align*}
\[\dd^\T[0^\T] = \dd[0] = 0 = 0^\T\]
\[\dd^\T[\varepsilon^\T(f)] = \dd[\varepsilon(f)] =  \varepsilon\left( \dd[f] \right) =  \varepsilon^\T\left( \dd^\T[f] \right) \]
\noindent {\bf [C$\dd$.2]}  $\dd^\T[f] \circ^\T \langle x, y +^\T z \rangle^\T =  \dd^\T[f] \circ^\T \langle x, y \rangle^\T + \dd^\T[f] \circ^\T \langle x + \varepsilon^\T(y), z \rangle^\T$ and $\dd^\T[f] \circ^\T \langle x, 0^\T \rangle^\T = 0^\T$ 
    \begin{align*}
\dd^\T[f] \circ^\T \pair{x}{y +^\T z}^\T &=~\mu \circ \dd\left[ \T(f) \right] \circ \pair{x}{y +^\T z} \\ 
&=~\mu \circ \dd\left[ \T(f) \right] \circ \pair{x}{y + z} \\ 
&=~\mu \circ \dd\left[ \T(f) \right] \circ \pair{x}{y} + \mu \circ \dd\left[ \T(f) \right] \circ \pair{x + \varepsilon(y)}{z} \tag{by {\bf [C$\dd$.2]}} \\
&=~ \mu \circ \dd\left[ \T(f) \right] \circ \pair{x}{y} + \mu \circ \dd\left[ \T(f) \right] \circ \pair{x +^\T \varepsilon^\T(y)}{z} \\
&=~ \dd^\T[f] \circ^\T \langle x, y \rangle^\T + \dd^\T[f] \circ^\T \langle x + \varepsilon^\T(y), z \rangle^\T \\\\
\dd^\T[f] \circ^\T \pair{x}{0^\T}^\T &=~ \dd^\T[f] \circ^\T \pair{x}{0}^\T \\
&=~\mu \circ \dd\left[ \T(f) \right] \circ \pair{x}{0} \\
&=~ 0 \tag{by {\bf [C$\dd$.2]}} \\
&=~ 0^\T
\end{align*}   
\noindent {\bf [C$\dd$.3]}   $\dd^\T[\eta] = \pi_1^\T$ and $\dd^\T[\pi^\T_i] =\pi_i^\T\circ^\T \pi_1^\T$ \\\\
First note that $\pi_i^\T = \eta \circ \pi_i$ and $\pi_i^\T\circ^\T \pi_1^\T = \eta \circ \pi_i \circ \pi_1$. Then since $\eta$ is linear, we have that: 
\begin{align*}
\dd^\T[\eta] &=~ \dd[\eta] \\
&=~ \eta \circ \pi_1 \tag{$\eta$ is linear} \\
&=~ \pi_1^\T \\\\
\dd^\T[\pi^\T_i] &=~ \dd[\pi^\T_i] \\
&=~ \dd[\eta \circ \pi_i] \\
&=~\eta \circ \dd[\pi_i] \tag{$\eta$ linear and Lem.\ref{Lemma:linear}.5} \\
&=~\eta \circ \pi_i \circ \pi_1 \tag{by {\bf [C$\dd$.3]}} 
\end{align*}
\noindent {\bf [C$\dd$.4]}  $\dd^\T[\langle f, g \rangle^\T] = \langle \dd^\T[f], \dd^\T[g] \rangle^\T$ \\\\
First note that $\phi$ is linear. Then we have that: 
 \begin{align*}
\dd^\T[\langle f, g \rangle^\T] &=~ \dd[\langle f, g \rangle^\T] \\
&=~ \dd\left[\phi \circ \langle f, g \rangle \right] \\
&=~ \phi \circ \dd\left[ \langle f, g \rangle \right] \tag{$\phi$ linear and Lem.\ref{Lemma:linear}.5} \\
&=~ \phi \circ \langle \dd[f], \dd[g] \rangle  \tag{by {\bf [C$\dd$.4]}} \\
&=~  \langle \dd[f], \dd[g] \rangle^\T \\
&=~ \langle \dd^\T[f], \dd^\T[g] \rangle^\T
\end{align*}

\noindent {\bf [C$\dd$.5]}  $\dd^\T[g \circ^\T f] = \dd^\T[g] \circ^\T \langle f \circ^\T \pi_0^\T, \dd^\T[f] \rangle^\T$ \\\\
First note that since $\pi_0^\T = \eta \circ \pi_0$, it easily follows that $f \circ^\T \pi_0^\T = f \circ \pi_0$ (using the monad identities and the naturality of $\eta$). Therefore, we compute that: 

\begin{align*}
 \dd^\T[g] \circ^\T \langle f \circ^\T \pi_0^\T, \dd^\T[f] \rangle^\T &=~ \mu \circ \dd\left[ \T(g) \right] \circ \langle f \circ^\T \pi_0^\T, \dd^\T[f] \rangle \\
 &=~\mu \circ \dd\left[ \T(g) \right] \circ \langle f \circ \pi_0, \dd[f] \rangle \\
    &=~ \mu \circ \dd[\T(g) \circ f]  \tag{by {\bf [C$\dd$.5]}} \\
    &=~ \dd\left[\mu \circ T(g) \circ f \right]  \tag{$\mu$ linear and Lem.\ref{Lemma:linear}.5} \\
    &=~\dd[g \circ^\T f] \\
    &=~ \dd^\T[g \circ^\T f] 
\end{align*}
For the remaining two axioms, we will instead prove {\bf [C$\dd$.6.a]} and {\bf [C$\dd$.7.a]}. Before we do so, we first compute the following: 
\begin{align*}
\dd^\T\left[\dd^\T[f] \right] \circ^\T \left \langle \langle  x, y \rangle^\T, \langle z,w \rangle^\T \right \rangle^\T &=~ \mu \circ \dd\left[ \T(\dd^\T[f]) \right] \circ  \left \langle \langle  x, y \rangle^\T, \langle z,w \rangle^\T \right \rangle \\
&=~ \mu \circ \dd\left[ \T(\dd[f]) \right] \circ  \left \langle \phi \circ \langle  x, y \rangle, \phi \circ \langle z,w \rangle \right \rangle \\
&=~ \mu \circ \dd\left[ \T(\dd[f]) \right] \circ  \left \langle \phi \circ \langle  x, y \rangle, \phi \circ \langle z,w \rangle \right \rangle \\
&=~ \mu \circ \dd\left[ \T(\dd[f]) \right] \circ (\phi \times \phi) \circ \four{x}{y}{z}{w} \\
&=~ \mu \circ \dd\left[ \T(\dd[f]) \circ \phi \right]\circ \four{x}{y}{z}{w} \tag{$\phi$ linear and Lem.\ref{Lemma:linear}.5} \\
&=~\mu \circ \dd^2\left[ \T(f) \right] \circ \four{x}{y}{z}{w} 
\end{align*}
Therefore we have that $\dd^\T\left[\dd^\T[f] \right] \circ^\T \left \langle \langle  x, y \rangle^\T, \langle z,w \rangle^\T \right \rangle^\T = \mu \circ \dd^2\left[ \T(f) \right] \circ \four{x}{y}{z}{w}$. \\\\
\noindent {\bf [C$\dd$.6.a]}  $\dd^\T\left[\dd^\T[f] \right] \circ^\T \left \langle \langle  x, 0^\T \rangle^\T, \langle 0^\T, y \rangle^\T \right \rangle^\T = \dd^\T[f] \circ^\T  \langle x, y \rangle^\T$
\begin{align*}
\dd^\T\left[\dd^\T[f] \right] \circ^\T \left \langle \langle  x, 0^\T \rangle^\T, \langle 0^\T, y \rangle^\T \right \rangle^\T &=~ \mu \circ \dd^2 \left[\T(f) \right] \circ \left \langle  \langle  x, 0^\T \rangle, \langle 0^\T,y \rangle \right \rangle \\
&=~  \mu \circ \dd^2 \left[\T(f) \right] \circ \left \langle  \langle  x, 0 \rangle, \langle 0,y \rangle \right \rangle \\
&=~\mu \circ \dd[\T(f)] \circ \pair{x}{y} \tag{by {\bf [C$\dd$.6]}} \\
&=~  \dd^\T[f] \circ^\T  \langle x, y \rangle^\T
\end{align*}
\noindent {\bf [C$\dd$.7.a]}  $\dd^\T\left[\dd^\T[f] \right] \circ^\T \left \langle \langle x, y \rangle^\T, \langle z, w \rangle^\T \right \rangle^\T= \dd^\T\left[\dd^\T[f] \right] \circ \left \langle \langle x, z \rangle^\T, \langle y, w \rangle^\T \right \rangle^\T$
\begin{align*}
\dd^\T\left[\dd^\T[f] \right] \circ^\T \left \langle \langle x, y \rangle^\T, \langle z, w \rangle^\T \right \rangle^\T &=~\mu \circ \dd^2\left[ \T(f) \right] \circ \four{x}{y}{z}{w} \\
&=~  \mu \circ \dd \left[ \dd[\T(f)] \right] \circ \left \langle  \langle  x, z \rangle, \langle y,w \rangle \right \rangle \tag{by {\bf [C$\dd$.7]}} \\
&=~  \dd^\T\left[\dd^\T[f] \right] \circ \left \langle \langle x, z \rangle^\T, \langle y, w \rangle^\T \right \rangle^\T 
\end{align*}
  So we conclude that the Kleisli category is a Cartesian difference category. 
  \hfill $\blacksquare$
\end{proof}

We also point that in the case of a Cartesian differential category, since $\varepsilon =0$, it follows that $\varepsilon^T = 0$ also. Therefore we have that the Kleisli category of a Cartesian differential category is again a Cartesian differential category. To the knowledge of the authors, this is a novel observation. 

\begin{corollary}  For a Cartesian differential category $\mathbb{X}$, the Kleisli category $\mathbb{X}_\mathsf{T}$ is a Cartesian differential category with differential combinator $\D^\T = \D$.
\end{corollary}

We conclude this section by taking a look at the linear maps and the $\varepsilon^\T$-linear maps in the Kleisli category. A Kleisli map $f=\langle f_0, f_1 \rangle$ is linear in the Kleisli category if $\dd^\T[f] = f \circ^\T \pi^\T_1$, which amounts to requiring that:
\[ \langle \dd[f_0], \dd[f_1] \rangle = \langle f_0 \circ \pi_1, f_1 \circ \pi_1 \rangle \]
Therefore a Kleisli map is linear in the Kleisli category if and only if it is the pairing of maps which are linear in the base category. Similarly, a Kleisli map is $\varepsilon^\T$-linear if and only if is the pairing of $\varepsilon$-linear maps. 

\section{Conclusions and Future Work}

We have presented Cartesian difference categories, which generalize Cartesian differential categories to account for more discrete definitions of derivative, while providing additional structure that is absent in change action models. We have also exhibited important examples and shown that Cartesian difference categories arise quite naturally from considering tangent bundles in any Cartesian differential category. Our claim is that Cartesian difference categories can facilitate the exploration of differentiation in discrete spaces, by generalizing techniques and ideas from the study of their differential counterparts.

For example, Cartesian differential categories can be extended to allow objects whose tangent space is not necessarily isomorphic to the object itself \cite{cruttwell2017cartesian}. The same generalization could be applied to Cartesian difference categories -- with some caveats: for example, the equation defining a linear map (Definition \ref{def:linearity}) becomes ill-typed, but the notion of $\varepsilon$-linear map remains meaningful.

Another relevant path to consider is developing the analogue of the ``tensor'' story for Cartesian difference categories. Indeed, an important source of examples of Cartesian differential categories are the coKleisli categories of a tensor differential category \cite{blute2006differential,blute2009cartesian}. It is likely that a similar result holds for a hypothetical ``tensor difference category'', but it is not clear how these should be defined: \CdCax{2} implies that derivatives in the difference sense are non-linear and therefore their interplay with the tensor structure will be much different.

A further generalization of Cartesian differential categories, categories with tangent structure \cite{cockett2014differential} are defined directly in terms of a tangent bundle functor rather than requiring that every tangent bundle be trivial (that is, in a tangent category it may not be the case that $\T A = A \times A$). Some preliminary research on change actions has already shown that, when generalized in this way, change actions are in fact precisely internal categories, but the consequences of this for change action models (and, \emph{a fortiori}, Cartesian difference categories) are not understood.

More recently, some work has emerged about differential equations using the language of tangent categories \cite{cockett2017connections}. We believe similar techniques can be applied in a straightforward way to Cartesian difference categories, where they might be of use to give an abstract formalization of discrete dynamical systems and difference equations.

Finally, an important open question is whether Cartesian difference categories (or a similar notion) admit an internal language. It is well-known that the differential $\lambda$-calculus can be interpreted in Cartesian closed differential categories \cite{manzonetto2012categorical}. Given their similarities, we believe there will be a very similar ``difference $\lambda$-calculus'' which could potentially have applications to automatic differentiation (change structures, a notion similar to change actions, have already been proposed as models of forward-mode automatic differentiation \cite{kelly2016evolving}, although work on the area seems to have stagnated).

 \bibliographystyle{splncs04}
 \bibliography{references}

\begin{thebibliography}{10}
\providecommand{\url}[1]{\texttt{#1}}
\providecommand{\urlprefix}{URL }
\providecommand{\doi}[1]{https://doi.org/#1}

\bibitem{alvarez2019fixing}
Alvarez-Picallo, M., Eyers-Taylor, A., Jones, M.P., Ong, C.H.L.: Fixing
  incremental computation. In: European Symposium on Programming. pp. 525--552.
  Springer (2019)

\bibitem{alvarez2020cartesian}
Alvarez-Picallo, M., Lemay, J.S.P.: Cartesian difference categories. In:
  International Conference on Foundations of Software Science and Computation
  Structures. pp. 57--76. Springer (2020)

\bibitem{alvarez2019change}
Alvarez-Picallo, M., Ong, C.H.L.: Change actions: models of generalised
  differentiation. In: International Conference on Foundations of Software
  Science and Computation Structures. pp. 45--61. Springer (2019)

\bibitem{blute2006differential}
Blute, R.F., Cockett, J.R.B., Seely, R.A.G.: Differential categories.
  Mathematical structures in computer science  \textbf{16}(06),  1049--1083
  (2006)

\bibitem{blute2009cartesian}
Blute, R.F., Cockett, J.R.B., Seely, R.A.G.: Cartesian differential categories.
  Theory and Applications of Categories  \textbf{22}(23),  622--672 (2009)

\bibitem{FMCS2018}
Bradet-Legris, J., Reid, H.: Differential forms in non-linear cartesian
  differential categories (2018), {F}oundational Methods in Computer Science

\bibitem{cai2014theory}
Cai, Y., Giarrusso, P.G., Rendel, T., Ostermann, K.: A theory of changes for
  higher-order languages: Incrementalizing $\lambda$-calculi by static
  differentiation. In: ACM SIGPLAN Notices. vol.~49, pp. 145--155. ACM (2014)

\bibitem{cockett2014differential}
Cockett, J.R.B., Cruttwell, G.S.H.: Differential structure, tangent structure,
  and sdg. Applied Categorical Structures  \textbf{22}(2),  331--417 (2014)

\bibitem{cockett2017connections}
Cockett, J., Cruttwell, G.: Connections in tangent categories. Theory and
  Applications of Categories  \textbf{32}(26),  835--888 (2017)

\bibitem{cruttwell2017cartesian}
Cruttwell, G.S.: Cartesian differential categories revisited. Mathematical
  Structures in Computer Science  \textbf{27}(1),  70--91 (2017)

\bibitem{ehrhard2003differential}
Ehrhard, T., Regnier, L.: The differential lambda-calculus. Theoretical
  Computer Science  \textbf{309}(1),  1--41 (2003)

\bibitem{ehrhard2018introduction}
Ehrhard, T.: An introduction to differential linear logic: proof-nets, models
  and antiderivatives. Mathematical Structures in Computer Science
  \textbf{28}(7),  995--1060 (2018)

\bibitem{kelly2016evolving}
Kelly, R., Pearlmutter, B.A., Siskind, J.M.: Evolving the incremental
  $\{$$\backslash$lambda$\}$ calculus into a model of forward automatic
  differentiation (ad). arXiv preprint arXiv:1611.03429  (2016)

\bibitem{kock2006synthetic}
Kock, A.: Synthetic differential geometry, vol.~333. Cambridge University Press
  (2006)

\bibitem{manzonetto2012categorical}
Manzonetto, G.: What is a categorical model of the differential and the
  resource $\lambda$-calculi? Mathematical Structures in Computer Science
  \textbf{22}(3),  451--520 (2012)

\bibitem{manzyuk2012tangent}
Manzyuk, O.: Tangent bundles in differential lambda-categories. arXiv preprint
  arXiv:1202.0411  (2012)

\bibitem{richardson1954introduction}
Richardson, C.H.: An introduction to the calculus of finite differences. Van
  Nostrand (1954)

\bibitem{steinbach2009boolean}
Steinbach, B., Posthoff, C.: Boolean differential calculus. In: Logic Functions
  and Equations, pp. 75--103. Springer (2009)

\end{thebibliography}
\end{document}